\documentclass{amsart}
\usepackage{graphicx}
\usepackage{graphicx}
\usepackage{amssymb}
\usepackage{amsfonts}
\usepackage{amsmath}
\usepackage{amsthm}
\newtheorem{thm}{Theorem}[section]

\newtheorem{lem}[thm]{Lemma}
\newtheorem{rem}[thm]{Remark}

\numberwithin{equation}{section}

\begin{document}
\title{Hedging of Game Options under Model Uncertainty in Discrete Time}
 \author{Yan Dolinsky\\
 Department of Statistics\\
Hebrew University, Jerusalem\\
 Israel }%

\address{
 Department of Statistics, Hebrew University, Mount Scopus,
Jerusalem 91905 \\
 {e.mail: yan.dolinsky@mail.huji.ac.il}}

 \date{\today}

\begin{abstract}
We introduce a setup of model uncertainty
in discrete time. In this setup we
derive dual expressions for the super--replication
prices of game options with upper semicontinuous payoffs.
We show that the super--replication price is equal to the supremum
over a special (non dominated) set of martingale measures,
of the corresponding Dynkin games values.
This type of results is also new for American options.

\end{abstract}
\maketitle \markboth{Y.Dolinsky}{Game Options under Model Uncertainty}
\renewcommand{\theequation}{\arabic{section}.\arabic{equation}}
\pagenumbering{arabic}

\section{Introduction}
\label{sec:1}
A game contingent claim (GCC) or game option, which was introduced in
\cite{Ki1}, is defined
as a contract between the seller and the buyer of the option such
that both have the right to exercise it at any time up to a
maturity date (horizon) $T$. If the buyer exercises the contract
at time $t$ then he receives the payment $Y_t$, but if the seller
exercises (cancels) the contract before the buyer then the latter
receives $X_t$. The difference $\Delta_t=X_t-Y_t$ is the penalty
which the seller pays to the buyer for the contract cancellation.
In short, if the seller will exercise at a stopping time
$\sigma\leq{T}$ and the buyer at a stopping time $\tau\leq{T}$
then the former pays to the latter the amount $H(\sigma,\tau)$
where
\begin{equation*}
H(\sigma,\tau)=X_{\sigma}\mathbb{I}_{\sigma<\tau}+Y_{\tau}\mathbb{I}_{\tau\leq{\sigma}}
\end{equation*}
and we set $\mathbb{I}_{Q}=1$ if an event $Q$ occurs and
$\mathbb{I}_{Q}=0$ if not.

A hedge (for the seller) against a GCC is defined as a pair
$(\pi,\sigma)$ that consists of a self financing strategy $\pi$
and a stopping time $\sigma$ which is the cancellation time
for the seller. A hedge is called perfect if no matter what exercise
time the buyer chooses, the seller can cover his liability to the
buyer.

Until now there is quite a good understanding
of pricing game options in the case where the probabilistic model is given. For details see
\cite{Ki2} and the references therein. However, so far super-replication of American options and game options
was not studied
in the case of volatility uncertainty.
In fact, super--replication under volatility uncertainty
was studied only for European options
(see, \cite{DM}, \cite{DM1}, \cite{N}, \cite{NS} and \cite{STZ}). In the papers
(see, \cite{DM1}, \cite{NS} and \cite{STZ})
the authors established a connection
between $G$--expectation which was introduced by Peng (see \cite{P1} and \cite{P2}), and
super--replication under volatility uncertainty
in continuous time models.

In this paper we introduce a discrete setup of volatility
uncertainty. We consider a simple model which consists of a savings
account and of one risky asset, and we assume that the payoffs are upper semicontinuous.
Our main result says that the super--replication price is equal to the supremum
over a special (non dominated)  set of martingale measures,
of the corresponding Dynkin games values.
In continuous time models,
the problem remains open for American options and game options.

Main results of this paper are formulated in the next section.
In Section 3 we
prove the main results of the paper for continuous payoffs.
This proof is quite elementary and does not use advanced tools.
In section 4 we
extend the main results
for upper semicontinuous payoffs.
This extension is technically involved and requires
the establishment of some stability results
for Dynkin games under weak convergence.

\section{Preliminaries and main results}
\label{sec:2}
\setcounter{equation}{0}
First we introduce a discrete time version of volatility uncertainty.
Let $N\in\mathbb{N}$, $s>0$
and $I=[a,b]\subset\mathbb{R}_{+}$.
Define the set $K\subset\mathbb{R}^{N+1}_{++}$ by
$$K=\{(x_0,...,x_N): x_0=s ,\ \ |\ln x_{i+1}- \ln x_i| \in I, \ \ i<N\}.$$

The  financial market consists of
 a savings account $B$ and
 a risky asset $S$ (stock).
 The stock price
 process is
 $S_k$, $k=0,1,...,N$,
 where $N<\infty$ is
the maturity date or the total number of allowed
trades.
 By discounting,
 we normalize $B\equiv 1$.
 We assume that the stock price
 process satisfies $(S_0,...,S_N)\in K$.
Namely the initial stock price is $S_0=s$ and for any $i<N$ we have
$|\ln S_{i+1}- \ln S_i|\in I$.
This is the only assumption that we make on our financial market
and we do not assume any probabilistic structure.

For any $k=0,1,...,N$ let $F_k, G_k:K\rightarrow\mathbb{R}_{+}$
be upper semicontinuous functions with the following
properties, for any $u,v\in K$, $F_k(u)=F_k(v)$
and $G_k(u)=G_k(v)$ if $u_i=v_i$ for all $i=0,1,...,k$.
Furthermore, we assume that $F_k\leq G_k$.

Consider a game option with the payoff
function
\begin{equation}\label{2.2}
\mathbb{H}(k,l,S)=G_k(S)\mathbb{I}_{k<l}+F_l(S)\mathbb{I}_{l\leq k}, \ \ k,l=0,1,...,N.
\end{equation}
Observe that $\mathbb{H}(k,l,S)$ is the reward
that the buyer receives given that his exercise time is $l$ and that the seller
cancelation time is $k$. Furthermore, the reward $\mathbb{H}(k,l,S)$
depends only on the stock history up to the moment $k\wedge l$.

In our setup a portfolio with initial capital $x$
is a pair $\pi=(x,\gamma)$ where
$\gamma:\{0,1,...,N-1\}\times K\rightarrow\mathbb{R}$
is a progressively measurable process, namely
for any $k=0,1,...,N-1$ and $u,v\in K$,
$\gamma(k,u)=\gamma(k,v)$
if $u_i=v_i$ for all $i=0,1,...,k$. The portfolio value
at time $k$ is given by
\begin{equation}\label{2.3}
V^\pi_k(S)=x+\sum_{i=0}^{k-1}\gamma(i,S)(S_{i+1}-S_i), \ \ S\in K, \ \ k=0,1,...,N.
\end{equation}

A stopping time is a measurable function $\sigma:K\rightarrow\{0,1,...,N\}$
which satisfies the following, for any $u\in K$ and $k=0,1,...,N$
if $\sigma(u)=k$ then $\sigma(v)=k$ for any $v$ with
$v_i=u_i$ for all $i=0,1,...,k$.

A pair $(\pi,\sigma)$ of a self financing
strategy $\pi$ and a stopping time
$\sigma$ will be called a hedge.
A hedge $(\pi,\sigma)$ will called perfect if
\begin{equation}\label{2.4}
 V^\pi_{\sigma(S)\wedge l}(S)\geq \mathbb{H}(\sigma(S),l,S), \ \ \forall S\in K, \ \ l=0,1,...,N.
\end{equation}

The super--replication price is given by
\begin{equation}\label{2.5}
\mathbb{V}=\inf\left\{V^\pi_0|\  \mbox{there} \ \mbox{exists} \ \mbox{a} \
\mbox{stopping} \ \mbox{time} \ \sigma \ \mbox{such} \ \mbox{that}
\ (\pi,\sigma) \ \mbox{is} \ \mbox{a} \ \mbox{perfect} \ \mbox{hedge}\right\}.
\end{equation}

Observe that we do not have any underlying probability measure, and we
require to construct a super--hedge for any possible
values of the stock prices. Similar setup (but not the same)
was studied in \cite{DM} for European options.

We make some preparations before we formulate the main result of the paper.
Let $Z=(Z_0,...,Z_N)$ be the canonical process on the Euclidean space $\mathbb{R}^{N+1}$.
Namely for any $z=(z_0,...,z_N)\in\mathbb{R}^{N+1}$ and $k\leq N$ we have $Z_k(z)=z_k$.
A probability measure $\mathbb P$ supported on $K$ is called a martingale law if for any  $k<N$
\begin{equation}\label{2.6}
\mathbb{E}_{\mathbb P}(Z_N|Z_0,...,Z_k)=Z_k \ \ \mathbb{P} \ \ \mbox{a.s.}
\end{equation}
 where $\mathbb{E}_{\mathbb P}$ denotes the expectation with respect to $\mathbb P$.
Denote by $\mathcal{M}$ the set of
all martingale laws.
Clearly, $\mathcal{M}\neq\emptyset$. For instance the probability measure $\mathbb{P}_b$
which is given by
\begin{eqnarray*}
&\mathbb{P}_b(Z_0=s)=1 \ \ \mbox{and} \\
&\mathbb{P}_b(\ln Z_{i+1}-\ln Z_i=b)=
1-\mathbb{P}_b(\ln Z_{i+1}-\ln Z_i=-b)=\frac{1-e^{-b}}{e^b-e^{-b}}, \ \ i<N,
\end{eqnarray*}
is an element in $\mathcal{M}$.

Let $\mathcal{F}_k=\sigma(Z_0,...,Z_k)$, $k\leq N$ be the canonical filtration,
and let $\mathcal T$ be the set of all stopping times (with respect to the above filtration)
with values in the set $\{0,1,...,N\}$.

The following theorem is the main result of the paper.
\begin{thm}\label{thm2.1}
The super--replication price is given by
\begin{eqnarray*}
&\mathbb V=\inf_{\sigma\in\mathcal{T}}
\sup_{\mathbb P\in \mathcal{M}}
\sup_{\tau\in\mathcal{T}}
\mathbb{E}_{\mathbb P}\mathbb{H}(\sigma,\tau,Z)=\\
&\sup_{\mathbb P\in \mathcal{M}}
\inf_{\sigma\in\mathcal{T}}
\sup_{\tau\in\mathcal{T}}
\mathbb{E}_{\mathbb P}\mathbb{H}(\sigma,\tau,Z)=
\sup_{\mathbb P\in \mathcal{M}}
\sup_{\tau\in\mathcal{T}}\inf_{\sigma\in\mathcal{T}}
\mathbb{E}_{\mathbb P}\mathbb{H}(\sigma,\tau,Z).
\end{eqnarray*}
\end{thm}

It is well known that $\inf\sup\geq \sup\inf$, thus in order to prove Theorem \ref{thm2.1}
it is sufficient to prove the following relations
\begin{equation}\label{2.7}
\mathbb V\leq
\sup_{\mathbb P\in \mathcal{M}}
\sup_{\tau\in\mathcal{T}}
\inf_{\sigma\in\mathcal{T}}
\mathbb{E}_{\mathbb P} \mathbb{H}(\sigma,\tau,Z)
\end{equation}
and
\begin{equation}\label{2.8}
\mathbb V\geq\inf_{\sigma\in\mathcal{T}}
\sup_{\mathbb P\in \mathcal{M}}
\sup_{\tau\in\mathcal{T}}
\mathbb{E}_{\mathbb P}\mathbb{H}(\sigma,\tau,Z).
\end{equation}

The first inequality is the difficult one and it will be
proved in Sections 3--4. The second inequality is simpler
and we show it by the following argument.

From (\ref{2.5}) it follows that for any $\epsilon>0$ there exists a perfect hedge
$(\tilde\pi,\tilde\sigma)$
with an initial capital $V^{\tilde\pi}_0=\mathbb V+\epsilon$.
From (\ref{2.3}) we get that for any $\mathbb P\in\mathcal{M}$ the stochastic process
${\{V^{\tilde\pi}_k(Z)\}}_{k=0}^N$ is a martingale with respect to $\mathbb P$. Observe that
$\tilde\sigma(Z)\in\mathcal{T}$, and so from
(\ref{2.4}) we obtain that for any $\tau\in\mathcal{T}$
\begin{equation*}
\mathbb V+\epsilon=V^{\tilde\pi}_0=\mathbb E_{\mathbb P} V^{\tilde\pi}_{\tilde\sigma(Z)\wedge\tau}\geq
\mathbb{E}_{\mathbb P} \mathbb{H}(\tilde\sigma(Z),\tau,Z).
\end{equation*}
The terms $\mathbb P\in\mathcal{M}$ and $\tau\in\mathcal{T}$ are arbitrary,
thus we conclude that
\begin{equation*}
\mathbb V+\epsilon\geq
\sup_{\mathbb P\in \mathcal{M}}
\sup_{\tau\in\mathcal{T}}\mathbb{E}_{\mathbb P} \mathbb{H}(\tilde\sigma(Z),\tau,Z)\geq
\inf_{\sigma\in\mathcal{T}}
\sup_{\mathbb P\in \mathcal{M}}
\sup_{\tau\in\mathcal{T}}
\mathbb{E}_{\mathbb P} \mathbb{H}(\sigma,\tau,Z).
\end{equation*}
By letting $\epsilon\downarrow 0$ we derive (\ref{2.8}).\qed

\begin{rem}
From Theorem 2.1 we obtain the following probabilistic corollary.
$$\inf_{\sigma\in\mathcal{T}}
\sup_{\mathbb P\in \mathcal{M}}
\sup_{\tau\in\mathcal{T}}
\mathbb{E}_{\mathbb P}\mathbb{H}(\sigma,\tau,Z)
=
\sup_{\mathbb P\in \mathcal{M}}
\sup_{\tau\in\mathcal{T}}\inf_{\sigma\in\mathcal{T}}
\mathbb{E}_{\mathbb P}\mathbb{H}(\sigma,\tau,Z).$$
This corollary is not obvious since the set $\mathcal M$ is a set of non dominated
probability measures, and so it does not follow from the results in \cite{KK}.
\end{rem}

\section{Proof of the main result}\label{sec:3}\setcounter{equation}{0}
\label{sec:3}\setcounter{equation}{0}
This section is devoted to the proof of
(\ref{2.7}), for the case where the functions
$F_k,G_k:K\rightarrow\mathbb{R}_{+}$, $k\leq N$
are continuous.

\subsection{Discretization of the space}
Let $n\in\mathbb N$.
Introduce the set
\begin{eqnarray*}
&K_n:=\{(x_0,...,x_N): x_0=s \ \ \mbox{and} \nonumber\\
&|\ln x_{i+1}- \ln x_i| \in \{a,a+(b-a)/n, a+2(b-a)/n,...,b\}\}.
\end{eqnarray*}
Consider a multinomial model
for which the stock price $S=(S_0,...,S_N)$ lies in the set $K_n$.
As before the savings account is given by $B\equiv 1$.
In this model a portfolio with an initial capital $x$ is a pair
$\pi=(x,\gamma)$ where $\gamma:\{0,1,...,N-1\}\times K_n\rightarrow\mathbb{R}$
is a progressively measurable process. A hedge is a pair $(\pi,\sigma)$
which consists of a portfolio strategy $\pi$
and a stopping time $\sigma$. A stopping time
is a map $\sigma:K_n\rightarrow\{0,1,...,N\}$
which satisfies that
if $\sigma(u)=k$ then $\sigma(v)=k$ for any $v$ with
$v_i=u_i$ for all $i=0,1,...,k$. A hedge $(\pi,\sigma)$ will called perfect
if
\begin{equation}\label{3.2}
 V^\pi_{\sigma(S)\wedge l}(S)\geq \mathbb{H}(\sigma(S),l,S), \ \ \forall S\in K_n, \ \ l=0,1,...,N
\end{equation}
where the portfolio value is given by the same formula as (\ref{2.3}).

Let
\begin{equation}\label{3.2+}
\mathbb{V}_n=\inf\left\{V^\pi_0|\  \mbox{there} \ \mbox{exists} \ \mbox{a} \
\mbox{stopping} \ \mbox{time} \ \sigma \ \mbox{such} \ \mbox{that}
\ (\pi,\sigma) \ \mbox{is} \ \mbox{a} \ \mbox{perfect} \ \mbox{hedge}\right\}
\end{equation}
be the super--replication price in the multinomial model.
Next, we introduce a modified super--replication price. Let $M>0$ and
let $\Gamma_M$ be the set of all portfolio strategies
$\pi=(x,\gamma)$ where $\gamma:\{0,1,...,N-1\}\times K_n\rightarrow[-M,M]$. Namely, we consider portfolios
for which the absolute value of the number of stocks is not exceeding $M$.
Consider the super--replication price
$$\mathbb{V}^M_n=\inf_{\pi\in\Gamma_M}\left\{V^\pi_0|\  \mbox{there} \ \mbox{exists} \ \mbox{a} \
\mbox{stopping} \ \mbox{time} \ \sigma \ \mbox{such} \ \mbox{that}
\ (\pi,\sigma) \ \mbox{is} \ \mbox{a} \ \mbox{perfect} \ \mbox{hedge}\right\}.$$

We will need the following technical lemma.
\begin{lem}\label{lem3.0}
There exists a constant $M>0$ (which is independent of $n$)
such that $$\mathbb{V}^M_n=\mathbb{V}_n.$$
\end{lem}
\begin{proof}
Clearly, $\mathbb{V}^M_n\geq\mathbb{V}_n.$  Thus its
sufficient to show that $\mathbb{V}^M_n\leq\mathbb{V}_n.$
Set
$$A=\max_{0\leq k\leq N}\sup_{x\in K}  F_k(x).$$
Clearly there exists a perfect hedge with an initial capital $A$ (in this case the investor
does not trade and stop only at the maturity). Let $(\pi,\sigma)$ be a perfect hedge in the sense
of (\ref{3.2}). We will assume (without loss of generality)
that the initial capital $V^\pi_0$ is no bigger than $A>0$.
Furthermore, since the option is exercised no later than in the moment $\sigma(S)$,
we can assume (without loss of generality)
 that $\gamma(k,S)\equiv 0$ for $k\geq \sigma(S)$.

First let us prove by induction that
for any $S\in K_n$ and $k=0,1,...,N$,
\begin{equation}\label{3.3}
V^\pi_{k\wedge\sigma(S)}(S)\leq A\left(1+\frac{e^b-1}{1-e^{-b}}\right)^{k} \  \ \mbox{and} \ \ |\gamma(k,S)|\leq \frac{A\left(1+\frac{e^b-1}{1-e^{-b}}\right)^{k}}{(1-e^{-b})S_k}.
\end{equation}
If $\sigma\equiv 0$ then the statement is clear. Thus we
assume that $\sigma(S)>0$ for any ($\sigma$ is a stopping time) $S\in K_n$.
Choose $S\in K_n$. Clearly, the portfolio value at time $1$ should be non
negative, for any possible growth rate of the stock. In particular we have,
$$V^\pi_{0}(S)+\gamma(0,S) s (e^b-1)\geq 0 \ \ \mbox{and} \ \ V^\pi_{0}(S)+\gamma(0,S) s (e^{-b}-1)\geq 0$$
and we conclude that $|\gamma(0,S)|\leq \frac{A}{s(1-e^{-b})}$. Thus (\ref{3.2}) holds for $k=0$.
Next, assume that $(\ref{3.3})$ holds for $k$, and we prove it for $k+1$.
From the induction assumption we get
\begin{eqnarray*}
&V^\pi_{(k+1)\wedge\sigma(S)}(S)= V^\pi_{k\wedge\sigma(S)}(S)+
\gamma(k\wedge\sigma(S),S)(S_{(k+1)\wedge\sigma(S)}- S_{k\wedge\sigma(S)})\leq\\
& A\left(1+\frac{e^b-1}{1-e^{-b}}\right)^{k}+
\frac{A\left(1+\frac{e^b-1}{1-e^{-b}}\right)^{k}}{(1-e^{-b})S_k}S_k(e^b-1)\leq
A\left(1+\frac{e^b-1}{1-e^{-b}}\right)^{k+1},
\end{eqnarray*}
as required. Next,
if $\sigma(S)\leq k+1$ then $\gamma(k+1,S)=0$. If $\sigma(S)>k+1$,
then the portfolio value at time $k+2$ should be non negative, for any possible growth rate of the stock.
Thus,
$$V^\pi_{k+1}(S)+\gamma(k+1,S) S_{k+1} (e^b-1)\geq 0 \ \ \mbox{and} \ \ V^\pi_{k+1}(S)+\gamma(k+1,S) S_{k+1} (e^{-b}-1)\geq 0$$
 and so,
 $$|\gamma(k+1,S)|\leq \frac{V^\pi_{k+1}(S)}{(1-e^{-b})S_{k+1}}\leq \frac{A\left(1+\frac{e^b-1}{1-e^{-b}}\right)^{k+1}}{(1-e^{-b})S_{k+1}}.$$
 This completes the proof of (\ref{3.3}).
Finally,
observe that $S_k\geq s e^{-b k}$ and so, we conclude that
for $M:=A s \frac{e^{bN}}{1-e^{-b}}\left(1+\frac{e^b-1}{1-e^{-b}}\right)^N$, we have
$|\gamma(k,S)|\leq M$ for all $k,S$.
 \end{proof}
Now, we can easily prove the following lemma.
\begin{lem}\label{lem3.1}
$$\mathbb V\leq \lim\inf_{n\rightarrow\infty}\mathbb{V}_n.$$
\end{lem}
\begin{proof}
Fix $\epsilon>0$. Let $n\in\mathbb{N}$.
Consider the multinomial model for which the stock price process
$S=(S_0,...,S_N)$ lies in the set $K_n$.
Let $(\pi,\sigma)$ be a perfect hedge for this multinomial model such that
$\pi=(\mathbb {V}_n+\epsilon,\gamma)$.
From lemma \ref{lem3.0} it follows that we can assume that
$|\gamma(k,S)|\leq M$ for any $k,S$.
Consider the map $\psi_n:K\rightarrow K_n$
which is given by $\psi_n(y_0,...,y_N)=(x_0,...,x_N)$
where
\begin{eqnarray*}
&x_0=y_0 \ \ \mbox{and} \ \ \mbox{for} \ \ k>0 \ \ \ln y_{i+1}=\ln y_i\\
& +sgn(\ln x_{i+1}-\ln x_i)(a+(b-a)[n(|\ln x_{i+1}-\ln x_i|-a)/(b-a)]/n)
\end{eqnarray*}
where $[v]$ is the integer part of $v$ and $sgn(v)=1$ for $v>0$
and $=-1$ otherwise.
For the original
financial market define a hedge $(\tilde\pi,\tilde\sigma)$ by the following
relations, $\tilde\pi=(\mathbb {V}_n+2\epsilon,\tilde\gamma)$
where
\begin{equation}\label{3.4}
\tilde\gamma(k,S)=\gamma(k,\psi_n(S)) \ \ \mbox{and} \ \ \tilde\sigma(S)=\sigma(\psi_n(S)) , \ \ k<N, \ S\in K.
\end{equation}
Observe that $\tilde\gamma$ is a progressively measurable map and
$\tilde\sigma$ is a stopping time. Thus $(\tilde\pi,\tilde\sigma)$ is indeed
a hedge for the original financial market.
From the continuity of the functions $F_k,G_k$, $k=0,1,...,N$
it follows that for sufficiently large $n$
\begin{equation}\label{3.5}
||S-\psi_n(S)||+|F_k(S)-F_k(\psi_n(S))|+|G_k(S)-G_k(\psi_n(S))|<\frac{\epsilon}{2 M N}, \ \ S\in K, k\leq N,
\end{equation}
where we denote
$||(z_0,...,z_N)||=\max_{0\leq i\leq N}|z_i|$.
Let $S\in K$. Set $Y^{(n)}=\psi_n(S)$.
From (\ref{3.5}) and the fact that $\gamma\in [-M,M]$
it follows that (for sufficiently large $n$)
 for any $l\leq N$ we get
\begin{eqnarray*}
&V^{\tilde\pi}_{l\wedge \tilde\sigma(S)}(S)=\epsilon+V^\pi_{l\wedge\sigma(Y^{(n)})}(Y^{(n)})+\\
&\sum_{k=0}^{l\wedge\tilde\sigma(S)-1} \gamma(k,Y^{(n)})((S_{k+1}-S_k)-(Y^{(n)}_{k+1}-Y^{(n)}_k))\geq \\
&\epsilon+\mathbb{H}(\sigma(Y^{(n)}),l,Y^{(n)})-2 N M||S-Y^{(n)}|\geq \mathbb{H}(\tilde\sigma(S),l,S).
\end{eqnarray*}
Thus for sufficiently large $n$, $\mathbb V\leq2\epsilon+\mathbb V_n$. Since $\epsilon>0$ was arbitrary
this concludes the proof.
\end{proof}

\subsection{Analysis of the multinomial models}
Fix $n\in\mathbb{N}$.
Let $\Omega=\mathbb{R}^{N+1}$. Define the piecewise
constant
stochastic processes
\begin{eqnarray*}
&S^{(n)}_t(z_0,...,z_N):=z_{[nt]}, \ \ Y^{(n)}_t(z_0,...,z_N)=F_{[nt]}(z_0,...,z_N) \\
&\mbox{and} \ \ X^{(n)}_t=G_{[nt]}(z_0,...,z_N), \ \ z\in \Omega, \ t\in [0,1].
\end{eqnarray*}
Let ${\{\mathcal{F}^{(n)}_t\}}_{t=0}^1$ be the filtration
which is generated by the process $S^{(n)}$.
The set $K_n\subset \Omega$ is finite, and so, there exists a probability measure $\mathbb P_n$ on
$\Omega$ which is supported on $K_n$ and gives
to any element in $K_n$ a positive probability.
Thus we can apply
Theorem 2.2
in \cite{KK} for
a market with one risky asset $S^{(n)}$ which lives on
the probability space $(\Omega,\mathcal{F}^{(n)}_1,\mathbb P_n)$,
and a game option with the payoffs $Y^{(n)}\leq X^{(n)}$. In this case the super--replication price
coincides with $\mathbb{V}_n$ which is given by (\ref{3.2+}).
Thus let $\mathcal{M}_n\subset \mathcal{M}$ be the set of
all martingale laws which are supported on the set $K_n$ and
$\mathbb{T}$ be the set of all stopping times (with respect
to the filtration ${\{\mathcal{F}^{(n)}_t\}}_{t=0}^1$)
with values in the set $[0,1]$.
From Theorem 2.2 in \cite{KK} and the fact that
the processes $Y^{(n)},X^{(n)}$ are piecewise constant we obtain
$$\mathbb V_n=\sup_{\mathbb P\in \mathcal{M}_n}
\sup_{\tau\in\mathbb{T}}\inf_{\sigma\in\mathbb{T}}
\mathbb{E}_{\mathbb P}(X^{(n)}_{\sigma}\mathbb{I}_{\sigma<\tau}+
Y^{(n)}_{\tau}\mathbb{I}_{\sigma\geq\tau})=\sup_{\mathbb P\in \mathcal{M}_n}
\sup_{\tau\in\mathcal{T}}\inf_{\sigma\in\mathcal{T}}
\mathbb{E}_{\mathbb P}\mathbb{H}(\sigma,\tau,Z).
$$
Since $\mathcal{M}_n\subset\mathcal M$, we conclude that for any $n\in\mathbb N$,
$$\mathbb V_n\leq \sup_{\mathbb P\in \mathcal{M}}
\sup_{\tau\in\mathcal{T}}\inf_{\sigma\in\mathcal{T}}
\mathbb{E}_{\mathbb P}\mathbb{H}(\sigma,\tau,Z).$$
This together with Lemma \ref{lem3.1} completes the proof
of (\ref{2.7}).\qed

\begin{rem}
An interesting question which remains open is
the limit behavior where the maturity date $N$ goes to infinity.
Namely, for a given $N\in\mathbb{N}$ consider
the interval $I:=I(N)= \left[\frac{a}{\sqrt N},\frac{b}{\sqrt N}\right]$.
Our conjecture is that for regular enough payoffs
the limit behavior of the super--replication prices $\mathbb V:=\mathbb{V}(N)$
as $N\rightarrow\infty$ is equal to
a stochastic game version of $G$--expectation,
defined on the canonical space $\mathcal{C}[0,T]$.
For European options the limit is the standard $G$--expectation, this follows from
\cite{D} and \cite{DNS}. It seems that the tool which was employed
in \cite{D} can work for the American options case. In this case the limit of the
super--replication prices is equal
to an optimal stopping version of $G$--expectation. However
for game options the problem is more complicated.
\end{rem}

\section{Extension for upper semicontinuous payoffs}\label{sec:4}\setcounter{equation}{0}
In this section we prove
(\ref{2.7}) for the case where
the functions $F_k, G_k:K\rightarrow\mathbb{R}_{+}$, $k\leq N$
are upper semicontinuous (and not necessarily continuous).

Let $\mathcal{A}=\max_{0\leq k\leq N}\sup_{x\in K} G_k(x)<\infty$.
By using similar arguments as in Lemma 5.3 in \cite{DS}
it follows that for any
$k=0,1,...,N$ there are two sequences
of continuous functions $\{F^{(n)}_k\}_{n=1}^\infty$
and $\{G^{(n)}_k\}_{n=1}^\infty$
which satisfy the following:\\
(i). $\mathcal{A} \geq G^{(n)}_k \geq G_k$, $\mathcal{A}\geq F^{(n)}_k \geq F_k$ and
$G^{(n)}_k\geq F^{(n)}_k$,
for all $n$.\\
(ii).
\begin{equation}\label{4.1}
\lim\sup_{n\rightarrow\infty}G^{(n)}_k(x_n)\leq G^{(n)}(x) \ \ \mbox{and} \ \ \lim\sup_{n\rightarrow\infty}F^{(n)}_k(x_n)\leq F(x)
\end{equation}
for every $x\in K$ and every
sequence  $\{x_n\}_{n=1}^\infty\subset K$
with $\lim_{n\rightarrow\infty} x_n=x$.\\
(iii). Furthermore, for any $n\in\mathbb{N}$ and $u,v\in K$, $F^{(n)}_k(u)=F^{(n)}_k(v)$
and $G^{(n)}_k(u)=G^{(n)}_k(v)$ if $u_i=v_i$ for all $i=0,1,...,k$.

Let $\mathbb{V}$
be the super--replication
price which corresponds to the payoff functions $F,G$,
and for any $n\in\mathbb{N}$ let $\mathbb{V}_n$
be the super--replication price which corresponds
to the payoff functions $F^{(n)},G^{(n)}$.

From (i), it follows that for any $n\in\mathbb{N}$, $\mathbb V\leq \mathbb V_n$.
Thus from Theorem 2.1 (for continuous payoffs) it follows
that
$$\mathbb V\leq \lim\inf_{n\rightarrow\infty}
\sup_{\mathbb P\in \mathcal{M}}
\sup_{\tau\in\mathcal{T}}\inf_{\sigma\in\mathcal{T}}
\mathbb{E}_{\mathbb P}\mathbb{H}^{(n)}(\sigma,\tau,Z)
$$
where
$$\mathbb{H}^{(n)}(k,l,S)=G^{(n)}_k(S)\mathbb{I}_{k<l}+F^{(n)}_l(S)\mathbb{I}_{l\leq k}, \ \ k,l=0,1,...,N, \ \ S\in K.$$
We conclude that in order to establish
(\ref{2.7})
we need to prove the following lemma.
\begin{lem}
$$\sup_{\mathbb P\in \mathcal{M}}
\sup_{\tau\in\mathcal{T}}\inf_{\sigma\in\mathcal{T}}
\mathbb{E}_{\mathbb P}\mathbb{H}(\sigma,\tau,Z)=\lim\inf_{n\rightarrow\infty}
\sup_{\mathbb P\in \mathcal{M}}
\sup_{\tau\in\mathcal{T}}\inf_{\sigma\in\mathcal{T}}
\mathbb{E}_{\mathbb P}\mathbb{H}^{(n)}(\sigma,\tau,Z).$$
\end{lem}
\begin{proof}
From (i), $$\sup_{\mathbb P\in \mathcal{M}}
\sup_{\tau\in\mathcal{T}}\inf_{\sigma\in\mathcal{T}}
\mathbb{E}_{\mathbb P}\mathbb{H}(\sigma,\tau,Z)\leq\lim\inf_{n\rightarrow\infty}
\sup_{\mathbb P\in \mathcal{M}}
\sup_{\tau\in\mathcal{T}}\inf_{\sigma\in\mathcal{T}}
\mathbb{E}_{\mathbb P}\mathbb{H}^{(n)}(\sigma,\tau,Z).$$
Thus we will prove that (infact this is the inequality that we need)
\begin{equation}\label{4.1+}
\sup_{\mathbb P\in \mathcal{M}}
\sup_{\tau\in\mathcal{T}}\inf_{\sigma\in\mathcal{T}}
\mathbb{E}_{\mathbb P}\mathbb{H}(\sigma,\tau,Z)\geq\lim\inf_{n\rightarrow\infty}
\sup_{\mathbb P\in \mathcal{M}}
\sup_{\tau\in\mathcal{T}}\inf_{\sigma\in\mathcal{T}}
\mathbb{E}_{\mathbb P}\mathbb{H}^{(n)}(\sigma,\tau,Z).
\end{equation}
For any $n\in\mathbb N$, let $\mathbb P_n\in \mathcal{M}$ and
$\rho_n\in\mathcal{T}$
be such that
\begin{equation}\label{4.2}
\sup_{\mathbb P\in \mathcal{M}}\sup_{\tau\in\mathcal{T}}
\inf_{\sigma\in\mathcal{T}}
\mathbb{E}_{\mathbb P}\mathbb{H}^{(n)}(\sigma,\tau,Z)<\frac{1}{n}+
\inf_{\sigma\in\mathcal{T}}
\mathbb{E}_{\mathbb P_n}\mathbb{H}^{(n)}(\sigma,\rho_n,Z).
\end{equation}
Consider the set $\Pi$ of all probability measures
on $K$, induced with the topology of weak convergence.
Observe that $\Pi$ is a compact set (this follows from Prohorov's theorem, see \cite{B}
Section 1
for details).
From the existence of the regular distribution function
(for details see \cite{S} page 227) we obtain
that there exist measurable functions
$h^{(n)}_k:\mathbb{R}^{k+1}\rightarrow \Pi$,
$k<N$,
such that for any Borel set $A\subset K$ and $n\in\mathbb {N}$
$$\mathbb P_n((Z_0,...,Z_N)\in A|Z_0,...,Z_k)=h^{(n)}_k(Z_0,...,Z_k)(A), \ \ \mathbb{P}_n \ \ \mbox{a.s}.$$
For any $n\in\mathbb{N}$ consider the distribution of (under the measure $\mathbb P_n$)
$$(\rho_n, Z_0,...,Z_N,h^{(n)}_0(Z_0),...,h^{(n)}_{N-1}(Z_0,...,Z_{N-1}))$$
on the space $[0,N]\times K\times \Pi^N $ with the product topology.

Since the space $[0,N]\times K\times \Pi^N$ is compact
then by Prohorov's theorem there is a subsequence which for simplicity we still denote by
$$(\rho_n,Z_0,...,Z_N,h^{(n)}_0(Z_0),...,h^{(n)}_{N-1}(Z_0,...,Z_{N-1})), \ \ n\in\mathbb{N}$$
which converges weakly.
Thus from the Skorohod representation theorem (see \cite{Du})
we obtain that
we can redefine the
sequence $$(\rho_n, Z_0,...,Z_N,h^{(n)}_0(Z_0),...,h^{(n)}_{N-1}(Z_0,...,Z_{N-1})), \ \ n\in\mathbb N$$
on a probability space $(\Omega,\mathcal{F},P)$ such that
we have $P$ a.s convergence
\begin{equation}\label{4.3}
(\rho_n, Z^{(n)}_0,...,Z^{(n)}_N,h^{(n)}_0(Z^{(n)}_0),...,h^{(n)}_{N-1}(Z^{(n)}_0,...,Z^{(n)}_{N-1}))\rightarrow (\rho, U_0,...,U_N,W_1,...,W_N).
\end{equation}
Redefining means that for any $n\in\mathbb{N}$
the distribution of $$(\rho_n,Z_0,...,Z_N,h^{(n)}_0(Z_0),...,h^{(n)}_{N-1}(Z_0,...,Z_{N-1}))$$
under $\mathbb P_n$ is equal to the distribution of
$$(\rho_n, Z^{(n)}_0,...,Z^{(n)}_N,h^{(n)}_0(Z^{(n)}_0),...,h^{(n)}_{N-1}(Z^{(n)}_0,...,Z^{(n)}_{N-1}))$$
under $P$.
Let $\mathcal{G}_k=\sigma\{U_0,...,U_k\}$, $k\leq N$
be the filtration which is generated by $U_0,...,U_N$.
Denote by $\mathcal{T}_U$ the set of all stopping times
(with respect to this filtration) with values
in the set $\{0,1,...,N\}$.
Next we show the following three properties:
\\
(I). The distribution of $(U_0,...,U_N)$ (on the space $K$)
is an element in $\mathcal{M}$.\\
(II). For any $k$, the conditional distribution
of $(U_0,...,U_N)$ given $U_0,...,U_k$ equals to
$W_k$.\\
(III). For any $k$,
the event $\{\tau=k\}$ and $\mathcal{G}_N$ are independent
given $\mathcal{G}_k$.

Denote by $E$ the expectation with respect to $P$. From
Lebesgue's dominated convergence theorem
if follows
that
for any $k\leq N$ and
continuous bounded functions
$f:\mathbb{R}^{k+1}\rightarrow\mathbb{R}$,
$g:K\rightarrow\mathbb{R}$ we have
\begin{eqnarray}\label{4.3+}
&E((U_N-U_k)f(U_0,...,U_k))=\\
&\lim_{n\rightarrow\infty}E((Z^{(n)}_N-Z^{(n)}_k)f(Z^{(n)}_0,...,Z^{(n)}_k))\nonumber\\
&=\lim_{n\rightarrow\infty}\mathbb{E}_{\mathbb{P}_n}((Z_N-Z_k)f(Z_0,...,Z_k))=0,\nonumber
\end{eqnarray}
where the last equality follows the fact that $\mathbb{P}_n\in\mathcal{M}$
is a martingale distribution.
From the definition of the topology on $\Pi$, we also have
\begin{eqnarray}\label{4.3++}
&E(f(U_0,...,U_k)g(U_0,...,U_N))=\\
&\lim_{n\rightarrow\infty}E(f(Z^{(n)}_0,...,Z^{(n)}_k)g(Z^{(n)}_0,...,Z^{(n)}_N))\nonumber\\
&=\lim_{n\rightarrow\infty}E(f(Z^{(n)}_0,...,Z^{(n)}_k)\int g(y) h^{(n)}_{k}(Z^{(n)}_0,...,Z^{(n)}_{k})(dy))\nonumber\\
&=E(f(U_0,...,U_k)\int g(y)W_k(dy)).\nonumber
\end{eqnarray}
By applying standard density arguments
we obtain that (\ref{4.3+}) implies (I)
and (\ref{4.3++}) implies (II).
Next, fix $k$. From (II) and the fact that $\rho_n$
is a stopping time we obtain
that
\begin{eqnarray*}
&E(\mathbb{I}_{\rho=k}E(g(U_0,...,U_N)|\mathcal{G}_k))=
E(\mathbb{I}_{\rho=k}\int g(y)W_k(dy))=\\
&\lim_{n\rightarrow\infty} E(\mathbb{I}_{\rho_n=k}\int g(y)h^{(n)}_k(Z^{(n)}_0,...,Z^{(n)}_k)(dy))=\\
&\lim_{n\rightarrow\infty}E(\mathbb{I}_{\rho_n=k}E(g(Z^{(n)}_0,...,Z^{(n)}_N)|Z^{(n)}_0,...,Z^{(n)}_k))=\\
&\lim_{n\rightarrow\infty}E(g(Z^{(n)}_0,...,Z^{(n)}_N)\mathbb{I}_{\rho_n=k})=E(g(U_0,...,U_N)\mathbb{I}_{\rho=k}),
\end{eqnarray*}
and again, from standard density
arguments we conclude that
$$E(\mathbb{I}_{\rho=k}|\mathcal{G}_k)=E(\mathbb{I}_{\rho=k}|\mathcal{G}_N)$$
and (III) follows.
Property (III) is important
because it implies
the following. For any
stochastic process $(L_0,...,L_N)$
which is adapted to the filtration $\mathcal{G}_k$, $k\leq N$,
we have
\begin{equation}\label{4.4}
EL_{\rho}\leq \sup_{\tau\in\mathcal{T}_U}EL_{\tau}.
\end{equation}
The proof of this implication can be done in the same way
as in Lemma 3.3 in \cite{D4}, and so we omit it.

Now we arrive to the final step of the proof.
Choose $0<\epsilon<1$.
Let $\tilde\sigma\in\mathcal{T}_U$ be such that
\begin{equation}\label{4.5}
\inf_{\sigma\in\mathcal{T}_U} E\mathbb{H}(\sigma,\rho,U)>E\mathbb{H}(\tilde\sigma,\rho,U)-\epsilon,
\end{equation}
where $U=(U_0,...,U_N)$.
For any $k$ there exists a continuous function
$f_k:\mathbb{R}^{k+1}\rightarrow\mathbb{R}$
such that $P(\mathbb{I}_{\tilde\sigma=k}\neq f_k(U_0,...,U_k))<\frac{\epsilon}{2^{k+1}}$.
For any $n\in\mathbb{N}$ define
$\tilde\sigma_n=N\wedge\min\{k|f_k(Z^{(n)}_0,...,Z^{(n)}_k)>\frac{1}{2}\}$.
Clearly $\tilde\sigma_n$ is a stopping time with respect to the filtration
generated by $Z^{(n)}_0,...,Z^{(n)}_N$.
Let $C$ be the following set
$$C=\{\omega\in\Omega | \exists m:=m(\omega) \ \ \mbox{such} \ \mbox{that} \ \forall n>m \ \tilde\sigma_n(\omega)=\tilde\sigma(\omega)\}.$$
From (\ref{4.3}) and the fact that $f_k$, $k\leq N$ are
continuous functions, it follows
that
\begin{equation}\label{4.6}
P(C)\geq 1-\sum_{i=0}^N \frac{\epsilon}{2^{i+1}}\geq 1-\epsilon.
\end{equation}
Observe that (\ref{4.3}) also implies that a.s. $\rho_n(\omega)=\rho(\omega)$ for sufficiently large $n$ (which depends on $\omega$).
Thus from property (ii) we get
$$\mathbb{H}(\tilde\sigma,\rho,U)\mathbb{I}_C\geq\lim\sup_{n\rightarrow\infty}\mathbb{H}^{(n)}(\tilde\sigma_n,\rho_n,Z^{(n)})\mathbb{I}_C$$
where $Z^{(n)}=(Z^{(n)}_0,...,Z^{(n)}_N)$. Since $\mathbb{H}$ and $\mathbb{H}^{(n)}$ are uniformly bounded by $\mathcal{A}$
then from Fatou's lemma we derive
\begin{equation}\label{4.7}
E\mathbb{H}(\tilde\sigma,\rho,U)\mathbb{I}_C\geq\lim\sup_{n\rightarrow\infty}E\mathbb{H}^{(n)}(\tilde\sigma_n,\rho_n,Z^{(n)})\mathbb{I}_C.
\end{equation}
Finally, let $\mathcal{Q}$ be the distribution of $(U_0,...,U_N)$. From (I) it follows
that $\mathcal{Q}\in\mathcal{M}$ is a martingale distribution.
 It is well known that for Dynkin games the $\inf$ and the $\sup$ can be exchanged (for details see \cite{O}).
 Thus from
(\ref{4.2})--(\ref{4.3}),
(\ref{4.4}) for $L_k=H(\sigma,k,U)$ and (\ref{4.5})--(\ref{4.7})
we get
\begin{eqnarray*}
&\sup_{\mathbb P\in \mathcal{M}}
\sup_{\tau\in\mathcal{T}}\inf_{\sigma\in\mathcal{T}}
\mathbb{E}_{\mathbb P}\mathbb{H}(\sigma,\tau,Z)\geq
\sup_{\tau\in\mathcal{T}_U}\inf_{\sigma\in\mathcal{T}_U}E\mathbb{H}(\sigma,\tau,U)\\
&=\inf_{\sigma\in\mathcal{T}_U}\sup_{\tau\in\mathcal{T}_U}E\mathbb{H}(\sigma,\tau,U)
\geq \inf_{\sigma\in\mathcal{T}_U}E\mathbb{H}(\sigma,\rho,U)\\
&\geq E\mathbb{H}(\tilde\sigma,\rho,U)\mathbb{I}_C-\epsilon
\geq\lim\sup_{n\rightarrow\infty}E\mathbb{H}^{(n)}(\tilde\sigma_n,\rho_n,Z^{(n)})\mathbb{I}_C-\epsilon\\
&\geq\lim\sup_{n\rightarrow\infty}E\mathbb{H}^{(n)}(\tilde\sigma_n,\rho_n,Z^{(n)})-\mathcal{A}\epsilon-\epsilon\\
&\geq\lim\sup_{n\rightarrow\infty}
\sup_{\mathbb P\in \mathcal{M}}
\sup_{\tau\in\mathcal{T}}\inf_{\sigma\in\mathcal{T}}
\mathbb{E}_{\mathbb P}\mathbb{H}^{(n)}(\sigma,\tau,Z)-\epsilon(\mathcal A+1),
\end{eqnarray*}
 and since $\epsilon$ was arbitrary we obtain (\ref{4.1+}) as required. The reason that we have $\lim\sup$ in the above equation
 and not $\lim\inf$ as in (\ref{4.1+}),
 is because we passed to a subsequence, but left the same notations.
 \end{proof}
\begin{rem}
 Let us notice that in order to obtain Lemma 4.1 we used a stronger form of
 the standard weak convergence. Namely we also required a convergence of the conditional
 distributions. This is the discrete analog of the extended weak convergence
 which introduced by Aldous in \cite{A} for continuous time processes.
 In general, the standard weak convergence is not sufficient
 for the convergence of the corresponding
 optimal stopping and Dynkin games values.
 \end{rem}

\bibliographystyle{spbasic}

{}

\end{document}